\documentclass{article} 

\usepackage[a4paper, total={7in, 10in}]{geometry}
\usepackage{lipsum}

\usepackage{biblatex}

\usepackage{amsthm}

\usepackage{amsmath,amssymb,amsfonts}
\usepackage{nameref}
\usepackage{textcomp}
\usepackage{algorithm}
\usepackage{xfrac}
\usepackage{faktor}
\usepackage{stmaryrd}
\usepackage{bm} 
\usepackage{hyperref} 
\newtheorem{Example}{Example}
\newtheorem{theorem}{Theorem}

\newtheorem{lemma}{Lemma}
\newtheorem{remark}{Remark}

\newcommand\blfootnote[1]{%
  \begingroup
  \renewcommand\thefootnote{}\footnotetext{#1}%
  \addtocounter{footnote}{-1}%
  \endgroup
}

\begin{document}

\title{\textbf{Certification of Linear Inclusions for Nonlinear Systems}}
\date{}

\maketitle

\begin{center}
\author{Yehia Abdelsalam,}
\author{Sebastian Engell.}
\end{center}

\blfootnote{The authors are with the Process Dynamics and Operations Group, Biochemical Engineering Department, TU Dortmund University, Dortmund, Germany.}
\blfootnote{Emails: yehia.abdelsalam@tu-dortmund.de,  sebastian.engell@tu-dortmund.de.}
\blfootnote{This work was supported by TU Dortmund.}

\abstract{In this work, we propose novel method for certifying if a given set of vertex linear systems constitute a linear difference inclusion for a nonlinear system. The method relies on formulating the verification of the inclusion as an optimization problem in a novel manner. The result is a Yes/No certificate. We illustrate how the method can be useful in obtaining less conservative linear enclosures for nonlinear systems. 
}
\section{Introduction}
Differential (or difference) inclusions describe the evolution of dynamic systems using set valued maps \cite{aubin1984differential}. They have played a significant role in analyzing solutions of differential (or difference) equations with discontinuous right hand sides \cite{Filipov,Filippov1988DifferentialEW,Angeli}.    

A special class of a difference inclusions is the Linear Difference Inclusion (LDI) which can be described as $x^+\in \mathcal{F}(x,u)=\{Ax+Bu, \forall (A,B)\in \mathbf{D} \}$, where $\mathbf{D}$ is a compact set. This description is important in control theory because it is often used to analyze robust stability and to design controllers and observers for nonlinear and uncertain systems \cite{Boyd:94,Blanchini_set_book,Guerra}. The main idea of these analysis and design methods which rely on LDIs has its roots in \cite{Lure} which has resulted in the development of absolute stability theory (see for example \cite{popov,pyatniskij_1}). Since then, and a lot of research has emerged on the stability and control of LDIs \cite{MOLCHANOV198959,DAAFOUZ2001355,HU2007685,HU2010190}. 

 Takagi-Sugeno and other linear parameter varying representations of nonlinear systems \cite{Takagi_first,SHAMMA1991559,APKARIAN19951251,TS_control,TS_control_2,mohammadpour2012control} are usually based on convex combinations of several vertex linear systems, which is often called a polytopic LDI. For a wide range of systems, it is possible to determine the vertex linear systems using the nonlinearity sector approach \cite{sector_non}. If it is not possible to apply the sector nonlinearity approach (see Example \ref{Example_2} below), the mean value theorem can be applied (see \cite{Boyd_convex_op,Blanchini_set_book} for example) which can be very conservative.

In this paper, we devise a simple and novel method for checking if a given set of vertex linear systems is an actual LDI for a nonlinear system or not. Our method can be useful in reducing the conservatism of the representation by a LDI as shown in Example \ref{Example_2}. 

\section{Notations}
The set of real numbers is denoted by $\mathbb{R}$. Let $a \in \mathbb{R}^n$ be a real valued vector of dimension $n$. The notation $a\geq0$, $a>0$, $a\leq 0$, $a<0$ denote element-wise inequalities. A vector which has all its elements equal to one is denoted by $\mathbf{1}$, where the dimension is inferred from the context. The convex hull operator is denoted by $Co(\cdot)$. A superscript $T$ indicates the transpose operation.
\section{Problem Description}
Consider a dynamic system described by
\begin{equation}\label{eqn_sys_nonlinear}
	x^+=f(x,u), 
\end{equation}
where, $x \in \mathbb{R}^{n_x}$ and $u \in \mathbb{R}^{n_u}$ are the state and input of the system, $f:\mathbb{R}^{n_x}\times \mathbb{R}^{n_u} \times \to \mathbb{R}^{n_x}$ is a nonlinear map and $x^+$ denotes the successor state. 

Let $(x_s,u_s)$ denote an equilibrium pair for the system, i.e., $x_s=f(x_s,u_s)$.

Define the deviations from the equilibrium 
\begin{equation*}
    \delta x=x-x_s,
\end{equation*}
\begin{equation*}
    \delta u=u-u_s,
\end{equation*}
\begin{equation*}
    \delta x^+=x^+-x_s=f(x,u)-f(x_s,u_s).
\end{equation*}

Consider a compact set 
\begin{equation*}
    \mathcal{Z}=\{(x,u)| Fx+Eu\leq \mathbf{1}\}
\end{equation*}
which contains $(x_s,u_s)$ in its interior.
We want to verify that a set of vertex matrices $(\Bar{A}_i,\Bar{B}_i)$, $i \in \{1,2,\dots,n_d\}$ qualify as a polytopic LDI for the nonlinear system around the equilibrium pair $(x_s,u_s)$ inside the set $\mathcal{Z}$, i.e., given an equilibrium pair $(x_s,u_s)$, we want to verify if 
\begin{equation*}
    \delta x^+ \in Co(\{ \bar{A}_i\delta x+\bar{B}_i \delta u, \textbf{ } i \in \{1,2,\dots,n_d\}), \textbf{ }\forall (x,u) \in \mathcal{Z}.
\end{equation*}

To this end, we need to prove that for each $(x,u) \in \mathcal{Z}$ there exists a state and input dependent vector
\begin{equation*}
 \alpha(x,u)=\begin{pmatrix}
     \alpha_{1}(x,u) &\alpha_{2}(x,u)& \dots &\alpha_{n_d}(x,u)
 \end{pmatrix}^T
 \end{equation*}
such that 
\begin{equation}\label{alpha_1}
    \alpha_{i}(x,u)\geq 0, \textbf{ }\forall i \in \{1,2,\dots,n_d\},
\end{equation}
\begin{equation}\label{alpha_2}
   \sum^{n_d}_{i=1}\alpha_{i}(x,u)=1 
\end{equation}
and
\begin{equation}\label{LDI_eqn_1}	\delta x^+=\sum^{n_d}_{i=1}\alpha_{i}(x,u)(\bar{A}_i\delta x+\bar{B}_i \delta u), \textbf{ } \forall (x,u)\in \mathcal{Z}.
\end{equation}

\begin{remark}
    We do not assume that $(0,0)$ is an equilibrium pair for the system. If the origin is not an equilibrium pair, then it is not possible to find a LDI such that $x^+=\sum^{n_d}_{i=1}\alpha_{i}(x,u)(\bar{A}_i x+\bar{B}_i u)$, while it can be still possible to find a LDI in the sense of \eqref{LDI_eqn_1}.
\end{remark}
 \section{Main Result}
 Note that \eqref{alpha_1}, \eqref{alpha_2} and \eqref{LDI_eqn_1} must hold $\forall (x,u)\in \mathcal{Z}$ which is not possible to check in the current form. 

Consider the following version of Farkas' Lemma (see \cite{Boyd_convex_op,Manga_book} for example).
\begin{lemma}\label{Farkas_lemma}
	Let $M \in \mathbb{R}^{n\times m}$ and $b\in \mathbb{R}^n$. Then exactly one of the following statements is true:
	\begin{itemize}
		\item[1.] There exists $z\in \mathbb{R}^m$ with $z\geq 0$ such that $Mz=b$.
		\item[2.] There exists $p\in  \mathbb{R}^n$ such that $M^Tp\geq 0$ and $b^Tp<0$.
	\end{itemize}
\end{lemma}

Let $y\in \mathbb{R}^{n_x+1}$. Consider the following optimization problem.
\begin{subequations}\label{LDI_optimization}
	\begin{alignat}{4}
		& \underset{x, u, y }{\text{minimize}} \begin{pmatrix}
			(f(x,u)-x_s)^T & 1
		\end{pmatrix}   y \label{Cost_op}\\
		&  \text{subject to:} \notag \\
		& (x,u)\in \mathcal{Z},\\
		& \begin{pmatrix}
			(\bar{A}_1 \delta x+ \bar{B}_1 \delta u)^T& 1 \\ 
			(\bar{A}_2 \delta x+ \bar{B}_2 \delta u)^T & 1 \\
			\vdots & \vdots\\
			(\bar{A}_{n_d} \delta x+ \bar{B}_{n_d} \delta u)^T & 1 
		\end{pmatrix} \begin{pmatrix}
			y_1 \\ y_2 \\ \vdots \\y_{n_x+1}
		\end{pmatrix}\geq \begin{pmatrix}
		0 \\ 0 \\ \vdots \\0
		\end{pmatrix}, \label{cons_compare_farkas}\\
        & \sum^{n_x+1}_{j=1}y^2_j\leq c, \label{norm_cons}
	\end{alignat}
\end{subequations}
where $y_1,y_2,\dots,y_{n_x+1}$ are the elements of the vector $y$, and $c$ is some positive scalar. 

Let $x^*$, $u^*$, $y^*$ denote the optimal solution of \eqref{LDI_optimization}.
 
 The following theorem summarizes our main result.\\
 
 \begin{theorem}
The matrices $(\bar{A}_i,\bar{B}_i)$, $i\in \{1,2,\dots,n_d \}$ define a LDI for the system \eqref{eqn_sys_nonlinear} around an equilibrium pair $(x_s,u_s)$ if and only if for any $c>0$, \eqref{LDI_optimization} has a non-negative global optimal value.
\end{theorem}
\begin{proof}
Equations \eqref{alpha_1}, \eqref{alpha_2}, \eqref{LDI_eqn_1} can be equivalently written as 
\begin{equation}\label{alpha_for_comp}
    \alpha(x,u)\geq 0,
\end{equation}
\begin{equation}\label{first_part_farkas}
	\begin{pmatrix}
		f(x,u)-x_s \\ 1
	\end{pmatrix}=  \begin{pmatrix}
		\bar{A}_1 \delta x+ \bar{B}_1 \delta u& \bar{A}_2 \delta x+ \bar{B}_2 \delta u &\dots& \bar{A}_{n_d}\delta x+ \bar{B}_{n_d} \delta u \\ 1 & 1 & \dots & 1
	\end{pmatrix}\alpha(x,u).
\end{equation}

 The proof then follows by assigning $z$, $M$, $b$ and $p$ from Lemma \ref{Farkas_lemma} as follows:
	\begin{equation*}
		z=\alpha(x,u),
	\end{equation*}
	\begin{equation*}
		M=\begin{pmatrix}
		\bar{A}_1 \delta x+ \bar{B}_1 \delta u& \bar{A}_2 \delta x+ \bar{B}_2 \delta u &\dots& \bar{A}_{n_d} \delta x+ \bar{B}_{n_d} \delta u \\ 1 & 1 & \dots & 1
		\end{pmatrix},
	\end{equation*}
	\begin{equation*}
		b=\begin{pmatrix}
			f(x,u)-x_s\\1
		\end{pmatrix}
	\end{equation*}
	\begin{equation*}
		p=y, 
	\end{equation*}
and comparing the first statement of Lemma \ref{Farkas_lemma} with \eqref{alpha_for_comp} and \eqref{first_part_farkas}, and the second statement of Lemma \ref{Farkas_lemma} with \eqref{Cost_op} and \eqref{cons_compare_farkas}.

A negative optimal value of $\eqref{LDI_optimization}$ implies that at $x^*$, $u^*$, the second statement of Lemma \ref{Farkas_lemma} holds with $p=y^*$. This means that first statement of Lemma \ref{Farkas_lemma} does not hold at $x^*$, $u^*$, i.e., there does not exist $\alpha(x^*,u^*)$ that satisfies \eqref{alpha_1}, \eqref{alpha_2} and \eqref{LDI_eqn_1} at $x=x^*$ and $u=u^*$. 

To see the converse, note that a non-negative global optimum implies that for all admissible $(x,u)\in \mathcal{Z}$, \eqref{alpha_1}, \eqref{alpha_2} and \eqref{LDI_eqn_1} hold (i.e., the first statement of Lemma \ref{Farkas_lemma} holds $\forall (x,u) \in \mathcal{Z}$), and as a result, $(\bar{A}_i,\bar{B}_i)$, $i \in \{1,2,\dots, n_d\}$ is a LDI for the nonlinear system.

The constraint \eqref{norm_cons} only restricts the norm of the vector $y$ and does not change the sign of the objective function value or the constraints. 
\end{proof}

\begin{remark}
    A feasible and not necessarily optimal solution of \eqref{LDI_optimization} which results in a negative value of the objective function suffices for falsifying that $(\Bar{A}_i,\Bar{B}_i)$ is a LDI for the nonlinear system around the equilibrium $(x_s,u_s)$. The converse is not true, i.e., a non-negative local optimum of \eqref{LDI_optimization} is not sufficient for certifying that $(\Bar{A}_i,\Bar{B}_i)$ is a LDI for the nonlinear system around the equilibrium $(x_s,u_s)$.
\end{remark}
\begin{remark}
    Eliminating the constraint \eqref{norm_cons} from the optimization \eqref{LDI_optimization} does not affect the theoretical validity of the result. The constraint \eqref{norm_cons} was introduced for the purpose of constraining the decision variables of the optimization to compact sets. This can be very helpful for the numerical convergence of the solvers.
\end{remark}
A disadvantage of the proposed certification method is that \eqref{LDI_optimization} is a non-convex optimization problem. For non-convex optimization, finding a global solution is very difficult. Two main streams of research and algorithms exist for global non-convex optimization \cite{book:opt,LOCATELLI2021100012}. These are the stochastic/heuristic approaches and the deterministic/exact approaches, where each stream has its own advantages and disadvantages. For a rigorous verification, a deterministic global solver like Baron \cite{Baron} must be used. If a non-global solver (Ipopt \cite{IPOPT} for example) is employed, a strictly speaking non-rigorous approach is to solve the optimization problem \eqref{LDI_optimization} for a large number of random initial values to reduce the chance that a local optimum for which \eqref{Cost_op} is negative is missed (for stochastic global algorithms the probability of finding the global optimum tends to one as the number of samples/regions is increased).
\section{Illustrative Examples}
\begin{Example}\label{Example_1}
    Consider the following nonlinear system
    \begin{equation*}
        x^+_1=x^3_1+x^2_2,
    \end{equation*}
    \begin{equation*}
        x^+_2=x_1x_2,
    \end{equation*}
    in the compact region $x_1,x_2 \in [-1,1]$.
    The origin is an equilibrium for this system. We want to determine a LDI for the system around the origin.
    The system can be written in the following matrix form:
    \begin{equation*}
        \begin{pmatrix}
            x^+_1 \\x^+_2
        \end{pmatrix}=\begin{pmatrix}
	x_1^2 & x_2 \\ \dfrac{x_2}{2} & \dfrac{x_1}{2}
\end{pmatrix}  \begin{pmatrix}
            x_1 \\x_2
        \end{pmatrix}\label{eqn_example_mat},
    \end{equation*}
    i.e, $x^+=A(x)x$.
   We try to find a LDI for the system by considering two uncertain parameters in $A(x)$. To this end, we define the uncertain parameters in $A(x)$ as $a_1=x_1$ and $a_2=x_2$ over the range $[-1,1]$. Hence we have the four vertex matrices; $\bar{A}_1=\begin{pmatrix}
			1 & -1 \\ -\dfrac{1}{2} & -\dfrac{1}{2}
		\end{pmatrix}$, $\bar{A}_2=\begin{pmatrix}
			1 & 1 \\ \dfrac{1}{2} & -\dfrac{1}{2}
		\end{pmatrix}$, $\bar{A}_3=\begin{pmatrix}
			1 & -1 \\ -\dfrac{1}{2} &\dfrac{1}{2}
		\end{pmatrix}$, $\bar{A}_4=\begin{pmatrix}
			1 & 1 \\ \dfrac{1}{2} &\dfrac{1}{2}
		\end{pmatrix}$ as a candidate LDI for the nonlinear system. We then solve \eqref{LDI_optimization} to check if indeed $\Bar{A}_i$, $i \in \{1,2,3,4\}$ is a LDI for the nonlinear system. The resulting optimal value is $-0.385$, which means that $\bar{A}_1$,$\bar{A}_2$,$\bar{A}_3$,$\bar{A}_4$ is not a LDI for the system.  Figures \ref{Conv_violation_good} and \ref{Conv_violation} show two solutions to \eqref{LDI_optimization} which have the same optimal value of $-0.385$, which illustrates that the optimal solution to \eqref{LDI_optimization} is not necessarily unique.
   \begin{figure}[H]
	\begin{center} 
\includegraphics[width=1\columnwidth]{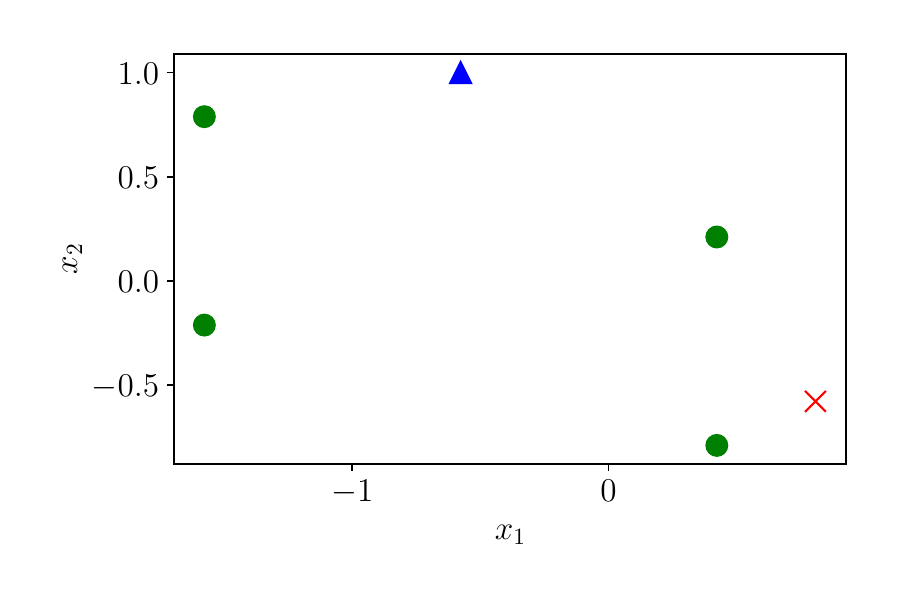} 
		\caption{Example \ref{Example_1}. The blue triangle is the state $x$. The red cross is the actual successor state $x^+$ that results from the nonlinear system equations. The green circles are $\bar{A}_ix$, $i=\{1,2,3,4\}$. The successor state $x^+\notin Co(\bar{A}_ix, i=\{1,2,3,4\})$. The optimal value corresponding to this solution is $-0.385$.}
		\label{Conv_violation_good}  
	\end{center}
\end{figure}	
\begin{figure}[H]
	\begin{center}
		\includegraphics[width=1\columnwidth]{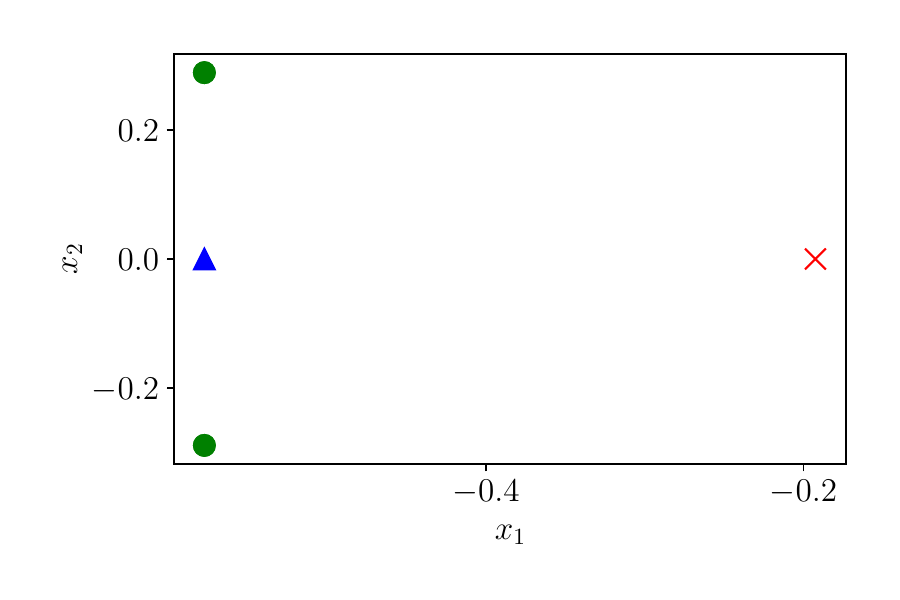} 
		\caption{Example \ref{Example_1} (another optimal solution). The blue triangle is the state $x$. The red cross is the actual successor state $x^+$ that results from the nonlinear system equations. The green circles are $\bar{A}_ix$, $i=\{1,2,3,4\}$. The successor state $x^+\notin Co(\bar{A}_ix, i=\{1,2,3,4\})$. The optimal value corresponding to this solution is $-0.385$.}
		\label{Conv_violation}  
	\end{center}
\end{figure}
The main reason that $\bar{A}_i$, $i=\{1,2,3,4\}$ is not a LDI for the nonlinear system is because the first element in $A(x)$ depends in a nonlinear fashion on $a_1$, and hence a third parameter $a_3\in [0,1]$ is needed for the LDI (to be used in the first element of $A(x)$). This results in a LDI with eight vertex matrices rather than four vertex matrices. Solving \eqref{LDI_optimization} for the new candidate LDI (with the eight vertex matrices) results in a positive optimal solution.  
\end{Example}
Note that in Example \ref{Example_1}, it was clear how to find the LDI using three parameters $a_1$, $a_2$, $a_3$. This is because the nonlinear system could be exactly represented as $x^+=A(x)x$, and $(0,0)$ was the equilibrium pair that we wanted to find the LDI around (see section 2.1.2 in \cite{Blanchini_set_book}). This is not the case for the next example.

\begin{Example}\label{Example_2}
Consider the nonlinear system
    \begin{equation*}
        x^+_1=2^{x_1}+x_2,
    \end{equation*}
    \begin{equation*}
        x^+_2=x_1+2x_2,
    \end{equation*}
    in the compact region $x_1,x_2 \in [-2,2]$. The origin is not an equilibrium point of this system. Note that it is not possible to use \cite{sector_non} to find a LDI for this system.
    
    One equilibrium for this system is $x_s=\begin{pmatrix}
        1 & -1
    \end{pmatrix}^T$. Using the mean value theorem (see Example 2.2 in \cite{Blanchini_set_book} or section 4.3 in \cite{Boyd:94}), we know that for each $x=\begin{pmatrix}
        x_1 & x_2
    \end{pmatrix}^T$, where $x_1,x_2 \in [-2,2]$, there exists $\hat{x}= \begin{pmatrix}
        \hat{x}_1 & \hat{x}_2
    \end{pmatrix}^T$, where $ \hat{x}_1,\hat{x}_2 \in [-2,2]$ such that 
    \begin{equation*}
        \delta x^+=  \frac{\partial f(x)}{\partial x} \vert_{\hat{x}} \delta x,
    \end{equation*}
    where 
    \begin{equation*}
        \delta x=x-x_s=\begin{pmatrix}
        x_1-1 & x_2+1
    \end{pmatrix}^T,
    \end{equation*}
     \begin{equation*}
         \delta x^+=x^+-x_s=\begin{pmatrix}
        x^+_1-1 & x_2^+ +1
    \end{pmatrix}^T,
     \end{equation*}
    \begin{equation*}
        \frac{\partial f(x)}{\partial x}=\begin{pmatrix}
            2^{x_1}\ln{(2)} & 1\\
            1 & 2
        \end{pmatrix}, 
    \end{equation*}
    where $\ln(\cdot)$ denotes the natural logarithm. Since $x_1\in [-2,2]$, $\Bar{A}_1=\begin{pmatrix}
        0.174 & 1\\ 1 &2
    \end{pmatrix}$ and $\Bar{A}_2=\begin{pmatrix}
        2.77 & 1 \\ 1 &2
    \end{pmatrix}$ define a LDI inclusion for the nonlinear system. As is well known, this method for the determination of LDIs is conservative. The conservatism of the obtained LDI can be reduced.
    A tighter candidate LDI is $\Bar{A}^{tight}_1=\begin{pmatrix}
        0.5 & 1\\ 1 &2
    \end{pmatrix}$ and $\Bar{A}^{tight}_2=\begin{pmatrix}
        2 & 1 \\ 1 &2
    \end{pmatrix}$. Solving \eqref{LDI_optimization} for this candidate LDI resulted in a non negative optimal value for $10^6$ random uniformly generated initial values of the optimization using the solver IPOPT \cite{IPOPT}, and hence $\Bar{A}^{tight}_2$ and $\Bar{A}^{tight}_2$ are considered to define a LDI for the nonlinear system.
    
    If we further tighten the LDI to $\Bar{A}^{n}_1=\begin{pmatrix}
        0.6 & 1\\ 1 &2
    \end{pmatrix}$ and $\Bar{A}^{n}_2=\begin{pmatrix}
        1.9 & 1 \\ 1 &2
    \end{pmatrix}$, the resulting optimal value for \eqref{LDI_optimization} becomes $-0.0499$ which means that $\Bar{A}^{n}_1$ and $\Bar{A}^{n}_2$ are not a LDI for the nonlinear system. Figure \ref{Conv_violation_EX2} illustrates the result of \eqref{LDI_optimization} with $\Bar{A}^{n}_1$ and $\Bar{A}^{n}_2$.

    \begin{figure}[H]
	\begin{center} 
\includegraphics[width=1\columnwidth]{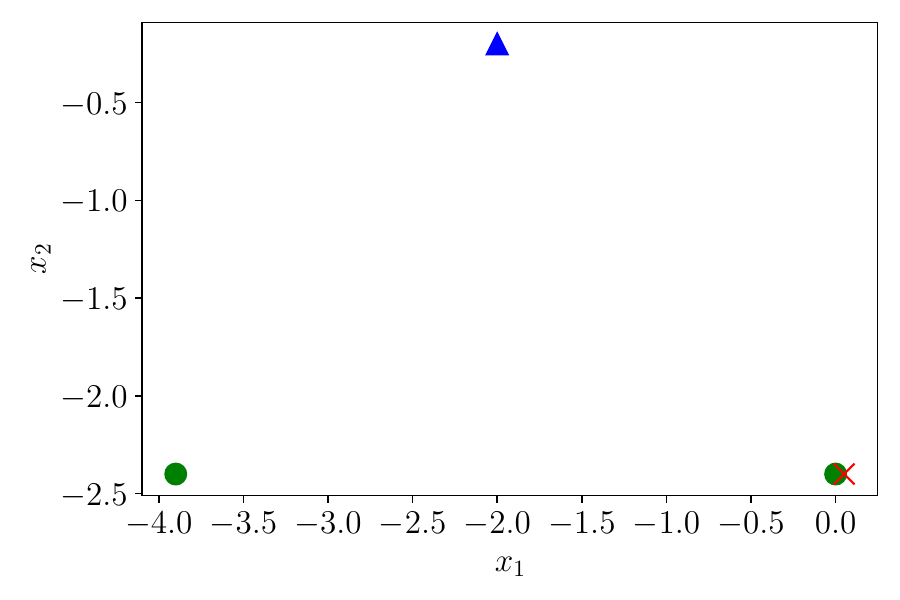} 
		\caption{Example \ref{Example_2} with $\Bar{A}^n_1$ and $\Bar{A}^n_2$. The blue triangle is the state $x$ which is $\begin{pmatrix}
		    -2 & -0.20001
		\end{pmatrix}^T$. The red cross is the actual successor state $x^+=\begin{pmatrix}
		    0.04999 & -2.40002
		\end{pmatrix}^T$ that results from the nonlinear system equations. The green circles are $\bar{A}^n_ix$, $i=\{1,2\}$. The right green circle is at $\begin{pmatrix}
		    -10^{-5} & -2.40002
		\end{pmatrix}^T$. The succesor state $x^+\notin Co(\bar{A}_ix, i=\{1,2\})$. The optimal value corresponding to this solution is $-0.0499$.}
		\label{Conv_violation_EX2}  
	\end{center}
\end{figure}

\end{Example} 
\section{Conclusion}
We have introduced a novel optimization problem which provides a Yes/No certificate on whether a candidate set of vertex linear systems is a LDI for a nonlinear system or not. Our result is valid for nonlinear systems which do not necessarily have the origin as an equilibrium. The proposed method can be used to reduce the conservatism in LDI determination. The benefits our method were illustrated by numerical examples. A disadvantage of the approach is that a non-convex nonlinear optimization problem has to be solved.

\printbibliography

\end{document}